\documentclass[prl,reprint,showpacs,superscriptaddress,a4paper]{revtex4-1}
\usepackage{amsmath, amsthm, amsfonts, mathrsfs}
\usepackage[breaklinks]{hyperref}
\usepackage{graphicx}

\newtheorem{proposition}{Proposition}
\newtheorem*{mainresult*}{Main result}

\DeclareMathOperator{\tr}{tr}
\DeclareMathOperator{\im}{im}
\DeclareMathOperator{\Prob}{Prob}
\newcommand{\abs}[1]{{\lvert{#1}\rvert}}
\newcommand{\norm}[1]{{\lVert{#1}\rVert}}
\newcommand{\reals}{{\mathbb R}}
\newcommand{\id}{{\openone}}
\newcommand{\econe}[1]{{C^+_{#1}}}
\newcommand{\eulere}{{\mathrm e}}

\begin{document}
\title{Typical local measurements in generalised probabilistic theories: emergence of quantum bipartite correlations}

\author{Matthias Kleinmann}
\email{matthias.kleinmann@uni-siegen.de}
\affiliation{%
Naturwissenschaftlich-Technische Fakult\"at,
Universit\"at Siegen,
Walter-Flex-Stra{\ss}e 3,
57068 Siegen, Germany}

\author{Tobias J.\ Osborne}
\email{tobias.osborne@itp.uni-hannover.de}
\affiliation{%
Institut f\"ur Theoretische Physik,
Leibniz Universit\"at Hannover,
Appelstra{\ss}e 2,
30167 Hannover, Germany}

\author{Volkher B.\ Scholz}
\email{scholz@phys.ethz.ch}
\affiliation{%
Institut f\"ur Theoretische Physik,
ETH Zurich,
Wolfgang-Pauli-Strasse 27,
8093 Zurich, Switzerland}
\affiliation{%
Institut f\"ur Theoretische Physik,
Leibniz Universit\"at Hannover,
Appelstra{\ss}e 2,
30167 Hannover, Germany}

\author{Albert H.\ Werner}
\email{albert.werner@itp.uni-hannover.de}
\affiliation{%
Institut f\"ur Theoretische Physik,
Leibniz Universit\"at Hannover,
Appelstra{\ss}e 2,
30167 Hannover, Germany}

\date{15 May 2012}

\begin{abstract}
	What singles out quantum mechanics as the fundamental theory of Nature? Here we study local measurements in generalised probabilistic theories (GPTs) and investigate how observational limitations affect the production of correlations. We find that if only a subset of \emph{typical} local measurements can be made then all the bipartite correlations produced in a GPT can be simulated to a high degree of accuracy by quantum mechanics. Our result makes use of a generalisation of Dvoretzky's theorem for GPTs. The tripartite correlations can go beyond those exhibited by quantum mechanics, however.
\end{abstract}


\maketitle

\emph{Introduction.}---%
The continued success of quantum mechanics (QM) strongly implies that it is the fundamental description of Nature.
However, it could still be that QM is simply a very good effective theory which
breaks down if we are able to perform experiments with sufficiently high energy
and precision. In this case QM would need to be replaced by a more general
``post-quantum''  theory. In particular \emph{generalised probabilistic
theories} (GPTs)  \cite{Barnum:2007PRL, Barnum:2011TCS, Barrett:2007PRA,
Masanes:2011NJP, Chiribella:2010PRA} have received considerable attention
recently, both as a foil to better understand the features of QM, and as a
powerful abstract way to reason about correlations and locality. These
investigations have lead to many interesting results, including simplified and
improved cryptographic schemes and primitives \cite{NoSigQKD,hanggi2010}.

If Nature is actually described by a theory other than QM then the natural question arises: why is QM such a good effective theory? A natural answer, which we investigate here, is that experimental imperfections prevent us from observing any post-quantum phenomena.

Suppose that Nature is described by a GPT with a high-dimensional state space  and corresponding high-dimensional set of all possible measurements. Observational limitations, such as detector resolution, mean that it is impossible to access most of these theoretically possible measurements. If physically implementable measurements are those chosen from some \emph{typical} subset (a precise definition is given in the sequel) then we show that the \emph{bipartite} correlations arising in any experiment can be modelled, to a high degree of precision, by those of QM. Note that the tripartite and multipartite correlations could go beyond those exhibited by QM: a sufficiently refined experiment involving three or more particles could exhibit behavior going beyond that possible within QM.

It is interesting to contrast our setting with that of \emph{decoherence}, which models the passage from the microscopic to the macroscopic \emph{classical} world \cite{Navascues:2010RSA, Kurzynski:2011XXX}. The crucial difference here is that decoherence arises from the correlations developed between a given particle and many other inaccessible particles (in the GPT framework it is rather likely that decoherence will always leads to an effective classical theory). By way of contrast, we consider only a few particles in isolation: roughly speaking, we study the case where only the ``local dimensions'' are effectively truncated.

Our argument builds on several important prior ideas. The first arises from the search \cite{ludwigbook, Mittelstaedt, Hardy:2001XXX,alfsen03} for an
 axiomatic derivation of QM: it was realised that a reasonable physical theory should
 allow for the convex combination of different possible measurements, and hence the underlying sets of both states and measurements should be \emph{dual} convex
 bodies. These developments have lead to the identification of generalised probabilistic theories as a general framework to study theories of physics going beyond QM.

The second cornerstone of our argument is the \emph{concentration of measure phenomenon} \cite{Milman:1986,
 Talagrand:1995IHE} epitomized by Dvoretzky's theorem which states, roughly,
 that a random low-dimensional section of a high-dimensional convex body looks
 approximately spherical.
This powerful result has already found myriad applications in quantum
 information theory, e.g., in quantum Shannon theory \cite{Hayden:2006CMP,
 Aubrun:2011CMP}, and quantum computational complexity theory
 \cite{Bremner:2009PRL, Gross:2009PRL}.
Here we adapt the ``tangible'' version of Dvoretzky's theorem for our purposes.

The final idea we exploit is the observation that \emph{spherical} state spaces can be simulated by \emph{sections} of quantum mechanical state spaces
 \cite{Tsirelson:1985ZNS}. As will become evident, our approach owes much to the recent work \cite{Barnum:2010PRL, Acin:2010PRL} showing that bipartite correlations may be modelled by QM when the constituents locally obey QM.

Here we exploit these three core ideas to obtain our
\begin{mainresult*}
If the local measurements in a GPT are chosen from a typical section of the convex body of all possible measurements then, with a high degree of accuracy, they do not yield any post-quantum prediction for the bipartite scenario.
\end{mainresult*}
More specifically, we require that the physically implementable measurements 
are in essence given by the section of the convex body of all measurements with 
a low-dimensional $\textsl{O}(n)$-\emph{typical} subspace. This means that the 
accessible measurements span a subspace and the choice of this subspace is not 
particular among all other subspaces of the same dimension.
This is a core assumption in our argument.  Although we restrict our attention 
 here to the case of a $\textsl{O}(n)$-typical subspaces, it is likely that our 
 result extends to a much wider variety of typicality notions.

Our argument
then implies that for most measurements given by low-dimensional subspaces
the outcomes can be explained using quantum mechanics. Hence we argue that 
those measurement devices revealing any post-quantum behavior are extremely 
difficult to build---since the choice of the right subspace requires extreme 
fine tuning.

\emph{Probabilistic physical theories, ordered vector spaces.}---%
It is useful to formulate GPTs in the mathematical language of \emph{ordered vector spaces} \cite{Alfsen:1971, Paulsen:2002, Barnum:2011TCS}: we begin with the description of the single-party state space and local measurements. The system is always assumed to be in a \emph{state} $\omega$, which encodes the probabilities of each outcome of all the possible measurements that may be performed. The set of all possible states, \emph{state space}, is denoted $\Omega$. Since any \emph{probabilistic} combination of states is, in principle, preparable, $\Omega$ is a convex set. We always assume that $\Omega$ is represented as a subset of $\mathbb{R}^n$.

A state $\omega\in \Omega$ assigns a probability to each \emph{outcome} of any 
possible measurement; a measurement outcome is represented by a map $f:\Omega 
\rightarrow [0,1]$. This map respects probabilistic mixtures of states, meaning 
that $f(p \omega_1 +(1-p)\omega_2) = p f(\omega_1)+(1-p)f(\omega_2)$. Extending 
each map linearly allows us to conclude that measurement outcomes are elements 
of the \emph{dual space} $V$ to $\mathbb{R}^n$. Any such $f$ is called an 
\emph{effect}. A special effect is the \emph{unit effect} $e$ defined by 
$e(\omega) = 1$ for all $\omega \in \Omega$. The unit effect represents a 
measurement with a single outcome: this is certain to occur regardless of what 
the state is. Convex combinations of effects are themselves assumed to be legal 
effects, so the set of effects is a convex subset of the \emph{dual} vector 
space $V$. A \emph{measurement} with $M$ outcomes is then a set of effects 
$\{f_j\}_{j=1}^M$ summing to the unit effect $e = \sum_{j=1}^M f_j$. This 
ensures that outcome probabilities of measurements sum to one. It is convenient 
to introduce the \emph{cone} generated by the zero effect, the unit effect, and 
all other effects, i.e., the set $V^{+} \equiv \{tf \, |\, t\ge 0, \text{$f$ is 
an effect}\}$.

The triple $(V, V^{+}, e)$ is known as an \emph{ordered unit vector space} and
encodes all of the theoretically possible local effects of a GPT. Throughout
the following we regard $(V, V^{+}, e)$ as the fundamental defining
representation of a GPT with state space as a derived concept (i.e., $\Omega$
is henceforth \emph{defined} as the set of all positive linear functionals
$\omega$ on $V$ such that $e(\omega) = 1$). It is convenient to assume a
further property, namely, that the triple $(V, V^{+},e)$ is \emph{Archimedean}.
This means that if $te + f \in V^{+}$ for all $t> 0$, then $f \in V^{+}$. Such
\emph{Archimedean ordered unit vector} spaces are referred to as \emph{AOU
spaces} in the sequel. The Archimedean axiom is a kind of closure assumption
which allows us, for example, to construct the order norm $\|f\|_+ \equiv \inf
\{t\,|\,te \pm f \in V^{+}, t\ge 0\}$. All ordered vector spaces can be
\emph{Archimedeanised} \cite{paulsen:2009a}, and from now on we assume that the
effects of a GPT are suitably represented by an AOU space.

An important example of a GPT is that of \emph{quantum mechanics} itself: an
$n$-level quantum system is described by an AOU space where $V\subset
M_{n}(\mathbb{C})$ is the set of  $n\!\times\! n$ hermitian matrices. The
effects are then the matrices $F\in V$ with $0\le F\le \id$ and the unit is
$e\equiv\id$. The cone $V^{+}$ generated by these effects is hence given by the
positive semidefinite matrices. One can verify that the triple $(V, V^{+}, e)$
is Archimedean. State space $\Omega$ is given by $\{F\mapsto \tr(\rho F)\, |\,
\rho\in V^{+}, \tr(\rho)=1\}$ and the order norm $\norm{A}_+$ is given by the
largest singular value of $A$.

\emph{Sections of GPTs.}---%
Here we study the effective theories arising from GPTs when only a subset of the possible effects may be implemented. For this purpose it is useful to introduce the notion of a linear map between AOU spaces: we say that a linear map $\varphi\colon V\rightarrow W$ between two AOU spaces $(V, V^+, e_V)$ and $(W, W^+, e_W)$ is \emph{positive} if $\varphi(V^+)\subset W^+$ and $\varphi$ is \emph{unital} when $\varphi(e_V)=e_W$.

Our definition of a \emph{section} of a GPT/AOU space $W$ is then motivated by
the observation that if we can only implement some subset of the effects in
$W^{+}$ then we can implement any convex combination of them. A particular
example of such a restriction is the \emph{intersection} of $W^{+}$ with some
subspace $V\subset W$. Since we can always apply the ``do nothing''
measurement, we require the subspace $V$ to contain $e_W$. Abstractly, a
section of $(W, W^{+}, e_W)$ is defined to be a positive unital injection
$\phi:V\hookrightarrow W$ such that $\phi(V^{+}) = W^{+}\cap \im\phi$. This 
last
condition has the consequence that the left inverse $\phi^{-1}$ is also a
positive unital linear map.

When restricted to a section of a GPT $(W, W^{+}, e_W)$ the state space of the
section $(V, V^{+}, e_V)$ is given by a \emph{quotient} of the state space of
$W$, i.e., $\Omega_V = \Omega_W/\sim$, where the equivalence relation is
determined by $\omega \sim \sigma$ if $f(\omega)=f(\sigma)$ for all $f\in V$.
This quotient is the \emph{shadow} of the convex body $\Omega_W$ on the
subspace $V$.

We now describe the AOU space playing the central role in our argument. This space is given by triple $(\mathbb{R}^{n+1},  \econe{n+1}(c), (1, \vec 0))$ where  $\econe{n+1}(c)$ denotes the $(n+1)$-dimensional Euclidean cone with length-diameter ratio $c:2$, i.e.,
\begin{equation}
 \econe{n+1}(c)= \{ (t, \vec x)\in \reals_+\times \reals^n \mid
 t\ge c \norm{\vec x}_2\},
\end{equation}
of which $e = (1, \vec 0)$ is the order unit.

It is a nontrivial fact that this space can be embedded into a quantum system, 
 i.e., it is a section of QM. The argument is due to Tsirelson 
 \cite{Tsirelson:1985ZNS} and proceeds as follows. Let $m=n/2$ if $n$ is even 
 and $m=(n+1)/2$ for odd $n$ and define $\gamma_1,
 \dotsc, \gamma_{2m}\in M_{2^m}(\mathbb{C})$ via $\gamma_{2j-1}= \sigma^{(1)}_z\dotsm
 \sigma^{(j-1)}_z \sigma_x^{(j)}$ and $\gamma_{2j  }= \sigma^{(1)}_z\dotsm
 \sigma^{(j-1)}_z \sigma_y^{(j)}$, where we've employed the standard Pauli matrix
 notation and juxtaposition indicates an implicit tensor product. Consider the 
positive unital injection
\begin{equation}
 \varphi\colon (t, \vec x)\mapsto t\id + c \sum_j x_j \gamma_j,
\end{equation}
(The positivity follows from $2 t\, \varphi(t, \vec x)=\varphi(t, \vec x)^2+
 (t^2- c^2\norm{x}_2^2)\, \id\ge 0$, arising from
 $\gamma_j\gamma_k + \gamma_k \gamma_j = 2\delta_{jk}\id$).
Since $\varphi$ is an injection, it has a left-inverse
\begin{equation}
 \varphi'\colon A\mapsto (\tr A, \tr(A \gamma_i)/c))/2^m,
\end{equation}
 which is again positive.
(Let $x_i \equiv \tr(A \gamma_i)$, so that $\tr (A) -c \norm{\vec x/c}_2 = \tr[A
 \varphi(1,-(\vec x/\norm{\vec x}_2)/c)]\ge 0$, since both matrices in the
 trace are already positive.)

\emph{Multipartite systems.}---%
We now discuss how to form joint systems in the GPT framework. Suppose Alice
and Bob are each in possession of a GPT $(V_A, V_A^{+}, e_A)$ and $(V_B,
V_B^{+}, e_B)$, respectively, which describes the purely \emph{local}
measurements for each party. The \emph{joint GPT} is then \emph{defined} to be
the AOU space $(V_A\otimes V_B, V_{AB}^+, e_A\otimes e_B)$ where, in order to
proceed, we must specify how to construct the cone $V_{AB}^+ \equiv
\text{``$(V_A\otimes V_B)^+$''}$. There are an infinite variety of
possibilities, however, we may restrict our attention to the following two
extremal definitions \cite{Han:2009XXX}. The first corresponds to the
\emph{maximal tensor product} $(V_A\otimes_{\text{max}} V_B)^+$ which is
defined to be the Archimedeanisation of the cone $\{ \sum_{j=1}^k f_j\otimes
g_j\,|\, f_j \in V^{+}_A,  g_j \in V^{+}_B, k \in \mathbb{N} \}$ and the second
to the \emph{minimal tensor product} $(V_A\otimes_{\text{min}} V_B)^+ \equiv \{
u\in V_A\otimes V_B\,|\, (\omega_A\otimes \omega_B)(u) \ge 0, \text{for all
$\omega_A\in \Omega_A$ and $\omega_B\in \Omega_B$} \}$.

By way of contrast, the tensor product used in the formation of joint systems in quantum mechanics is neither the minimal nor maximal one, but is rather strictly in between: $(V_A\otimes_{\text{max}} V_B)^+ \subset (V_A\otimes_{\text{QM}} V_B)^+ \subset (V_A\otimes_{\text{min}} V_B)^+$. The quantum mechanical tensor cone $V_{AB}^+$ is given by the set of positive semidefinite operators in $M_{n_A}(\mathbb{C})\otimes M_{n_B}(\mathbb{C})$. The state space $\Omega_{AB}^{\text{min}}$ corresponding to $(V_A\otimes_{\text{min}} V_B)^+$ is precisely the set of \emph{separable states} and the state space $\Omega_{AB}^{\text{max}}$ corresponding to $(V_A\otimes_{\text{max}} V_B)^+$ is given by the set of all positive semidefinite operators $W$ with $\tr(W) = 1$ which satisfy $\tr(W A\otimes B) \ge 0$, $\forall A,B \ge 0$. This set is dual to the set of \emph{entanglement witnesses} \cite{Horodecki:1996PLA} and includes all legal density operators as well as some operators with negative eigenvalues. Even though the state space $\Omega_{AB}^{\text{max}}$ in the case where our local GPTs are QM is strictly larger than quantum mechanical state space, results of \cite{Barnum:2010PRL, Acin:2010PRL} show that it does not give rise to any bipartite correlations going beyond QM. The following proposition is a slight generalization of this statement, dealing with (local) sections of quantum systems.

\begin{proposition}\label{p14427}
Consider two AOU spaces $(V_A, V^\mathrm{+}_A, e_A)$ and $(V_B, V^\mathrm{+}_B,
e_B)$ which are sections of quantum systems with according positive unital
injections $\varphi_A$ and $\varphi_B$ into an $n_A$-level (respectively,
$n_B$-level) quantum system. Assume, without loss of generality, that
$n_{A}\le n_{B}$. Then for any positive unital bilinear map $\omega_{AB}\colon
V_A\times V_B\rightarrow \reals$ there exists a state $\sigma_{AB}$ of the
composite quantum system $AB$ and a positive unital
 automorphism $\psi$ on $B$ such that $\omega_{AB}(f,g) = \tr(\sigma_{AB}\
\varphi_{A}(f)\otimes (\psi\circ\varphi_{B})(g))$.
\end{proposition}
\begin{proof}
By assumption the map $\omega_{AB}' (M_A, M_B)\mapsto 
 \omega_{AB}(\varphi^{-1}_{A}(M_A), \varphi^{-1}_{B}(M_B))$ is positive and 
 unital on the quantum systems $A$, $B$.
Hence the statement reduces to the case where $\varphi_{A}$ and $\varphi_B$ are both
 the identity mapping. A proof for this case was given by Barnum et al.\ \cite{Barnum:2010PRL}.
\end{proof}

We stress that the existence of positive unital left inverse maps
$\varphi_A^{-1}$ and $\varphi_B^{-1}$ is essential for this result to hold.
Indeed, in the case of a hypothetical nonlocal box \cite{Popescu:1994Found}, it
is impossible to find positive unital maps into quantum such that there left
inverse is also positive and hence non-local boxes allow post-quantum
behavior. It is also important to note that Proposition~\ref{p14427} does not
generalize to more than two parties \cite{Acin:2010PRL}.

\emph{Typical sections, main result.}---%
Consider an arbitrary pair of $n$-dimensional GPTs $A$ and $B$ and suppose that
we are only able to access a \emph{typical section} of the set of local effects
for $A$ (respectively, $B$).  This is modelled by the intersection of $V^{+}_A$
(respectively, $V^{+}_B$) with a typical $k$-dimensional subspace,
$k\ll n$. To do this abstractly we choose a bijection $T$ between $V$ and
$\mathbb{R}^n$ and
consider a random linear injection $X\colon
\reals^k\hookrightarrow \reals^n$ such that the random variable $X(\vec x)$ is
distributed according to the uniform measure on the Euclidean $(n-1)$-sphere of
radius $\norm{\vec x}_2$. (That is, $X$ is an $\textsl{O}(n)$-random rotation
of an embedded fiducial $k$-dimensional subspace.) We call
\begin{equation}
 Q(t, \vec x)=te + TX(\vec x)
\end{equation}
a \emph{centered random section} of $\reals^{k+1}$ into $V$ and it ensures that
every subspace corresponding to a typical choice of measurement settings
contains the neutral effect $e$. Since only convex combinations of $e$ with $TX(\reals^k)$ are feasible, we now study the cone $V^+\cap Q(\reals_+, \reals^k)$.

The following result captures the \emph{concentration of measure} phenomenon
for our setting.
\begin{proposition}\label{p20790}
Let $(V, V^+, e)$ be an $n$-dimensional AOU space and $0<\varepsilon<1$.
Then for $k\le \mathscr O(\varepsilon^2\log n )$ there exists a $k+1$ 
dimensional centered random section $Q$ of $V$, such that, with high 
probability,
\begin{equation}
 Q(\econe{k+1}(1+ \varepsilon)) \subset
 V^+\cap Q(\reals_+, \reals^k) \subset
  Q(\econe{k+1}(1-\varepsilon)).
\end{equation}
\end{proposition}
\begin{proof}
At the heart of the proof is the following ``tangible'' version of Dvoretzky's
 theorem \cite{Aubrun:2011CMP, Milman:1986, Pisier:1989}:
If $\eta\colon S^{n-1}\rightarrow \reals$ is a Lipschitz function with constant
 $L$ and central value $1$ (with respect to the uniform spherical measure on
 $S^{n-1}$), then for every $\varepsilon>0$, if $E\subset
 \reals^n$ is a random subspace of dimension $k\le k_0= c_0
 \varepsilon^2\, n/L^2$, we have, that
\begin{equation}
 \Prob\left[ \sup_{S^{n-1}\cap E} \abs{\eta(\vec x)- 1} > \varepsilon \right]
   \le c_1 \eulere^{-c_2 k_0},
\end{equation}
 where $c_0$, $c_1$, and $c_2$ are absolute constants.

For our scenario, we use $\eta(\vec z) = \inf\{t> 0\mid t e + T\vec z\in
V^+\}$ with $T$ chosen such that $\eta$  has a mean (which is a particular
central value) of $1$ on the $(n-1)$-dimensional Euclidean sphere and that the 
Lipschitz constant $L$ of $\eta$ is bounded via $L\le c' \sqrt{n/\log n}$ for 
some absolute constant $c'$. This is always possible, as can be seen following 
the proof of Theorem {4.3} in \cite{Pisier:1989}: First, by a Lemma of 
Dvoretzky and Rogers \cite[Theorem~{3.4}]{Milman:1986}, the bijection $T$ can 
be chosen such that for all canonical vectors $\vec e_k$ with $k\le n/2$ it 
holds that $\norm{T\vec e_k}_+\ge \norm T/4$. Without loss of generality we may 
assume in addition that $\eta$ has mean $1$.
Then, for a vector of normal distributed variables $\vec g$ and due to 
 $\norm{T\vec z}_+=\max\{\eta(\vec z),\eta(-\vec z)\}$ and
 \cite[Eqns.~(4.14, 4.18)]{Pisier:1989} we find,
\begin{multline}
2\sqrt n \ge 2 \mathbb E \eta(\vec g)
 \ge \mathbb E \norm{T\vec g}_+
 \ge \mathbb E \max_k \abs{g_k}\, \norm{T\vec e_k}_+ \\
 \ge \mathbb E \max_{k\le n/2} \abs{g_k}\, \norm{T\vec e_k}_+
 \ge c'' \sqrt{\log(n/2)} \, \norm T/4.
\end{multline}
On the other hand, $\eta$ is a sublinear function and thus
\begin{multline}
\abs{\eta(\vec z_1)-\eta(\vec z_2)}
 \le \max\{\eta(\vec z_1-\vec z_2), \eta(\vec z_2-\vec z_1)\}\\
 =   \norm{T(\vec z_1-\vec z_2)}_+
 \le \norm T \, \norm{\vec z_1-\vec z_2}_2,
\end{multline}
 which eventually shows $L\le c'\sqrt{n/\log n}$.

Now, by virtue of Dvoretzky's theorem, the following holds with high
 probability.
For all $\vec x\ne 0$ with $\xi\equiv \norm{\vec x}_2\le 1/(1 +\varepsilon)$,
 we have $\eta[X(\vec x/\xi)]\le 1 + \varepsilon\le 1/\xi$, and hence
 $Q(1, \vec x)= [e/\xi + TX(\vec x/\xi)]\xi \in V^+$.
Conversely, for all $\vec x$ with $\xi\equiv \norm{\vec x}_2>1/(1
 -\varepsilon)$, we have  $\eta[X(\vec x/\xi)]\ge 1- \varepsilon>1/\xi$,
 i.e., $Q(1, \vec x) \notin V^+$. The converse statement completes the proof.
\end{proof}

\begin{figure}
\includegraphics[width=.9\columnwidth]{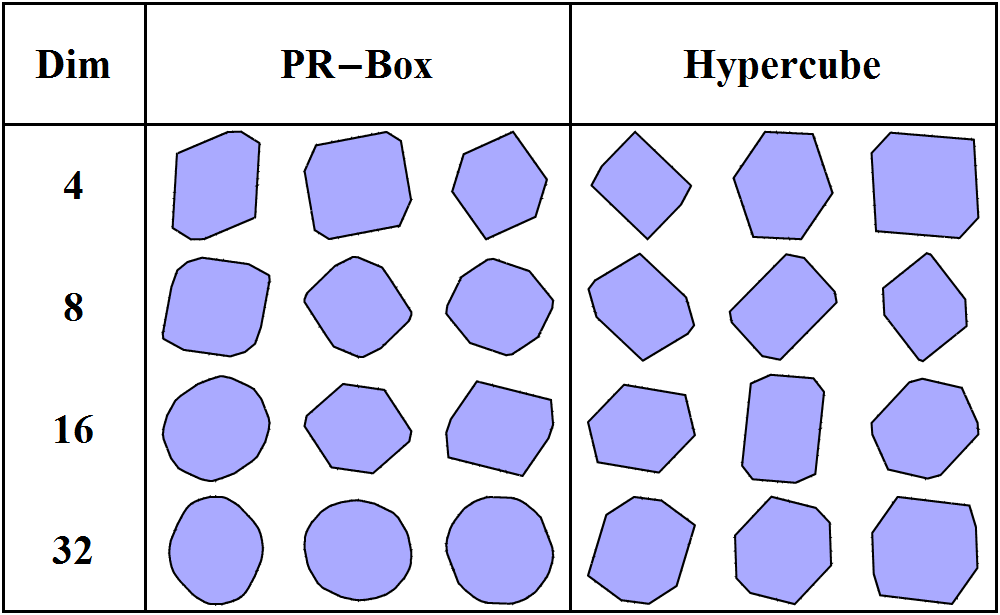}
\caption{Typical two-dimensional sections of a hypercube and of the effect space of a PR-Box
 in various dimensions. In both cases an increasing rounding of the corners of the sections can be observed.
 However in the case of a hypercube, which is the extremal situation for Dvoretzky's theorem, there is still an appreciable probability for non-rounded sections, due to low dimensionality.}
\label{fcut}
\end{figure}

Thus, with high accuracy, the effective theory corresponding to a
low-dimensional $\textsl{O}(n)$-typical section of a local GPT looks like a
Euclidean AOU space, cf.\ Fig.~\ref{fcut} for an illustration. The cones 
$Q(C_{k+1}^+(1\pm\varepsilon))$ give a very accurate description of the typical 
section, since by linearity all observable probabilities may at most deviate by 
$\mathscr
O(\varepsilon)$. Combining this with our previous finding, namely that 
Euclidean cones are sections of QM, and
hence, in view of Proposition~\ref{p14427}, all bipartite correlations of their
maximal tensor product may be simulated within QM, we arrive at our anticipated
main result.
Conversely, due to an argument by Tsirelson \cite{Tsirelson:1985ZNS}, all 
bipartite dichotomic correlations can be explained within an Euclidean cone of 
appropriate dimension.
Our result reduces to this dichotomic case, since already our description of a 
GPT by an AOU space is essentially limited to the dichotomic case.

Finally we briefly discuss the situation of a generalized Popescu-Rohrlich (PR)
 box, which exhibits (in some sense) the ``maximal'' possible post-quantum correlations \cite{Popescu:1994Found}.
Such boxes are locally described by an AOU vector space over $\reals^n$ with
 cone
 $\mathrm{PR}^+ = \{ (t,\vec x)\mid t\ge \sum_i \abs{x_i} \}$ and neutral
 element $(1,\vec 0)$.
By virtue of Proposition~\ref{p20790}, the fraction of 3-dimensional sections 
 from a $55\times10^6$-dimensional box with a post-quantum behavior of more 
 than $\pm 3\%$ is as low as $10^{-6}$\footnote{Those estimates stem from 
 estimating the constants $c_0$, $c_1$, and $c_2$ in the original proofs.  
 Numerically estimations yield much better bounds and the dimension can be 
 estimated to be $\approx 2000$ .}.

\emph{Conclusions.}---%
We have presented a mechanism whereby observable bipartite correlations of an arbitrary post-quantum theory could be, with high accuracy, compatible with those exhibited by  quantum mechanics. Our argument exploited the concentration of measure phenomenon and hence works for any typical low-dimensional section of a generalised probabilistic theory. We argued that such typical sections arise due to a lack of ultra-precise experimental control, in which case it would be virtually impossible to observe any post-quantum behavior, even if the fundamental theory of Nature wasn't quantum mechanics. This is complementary to the emergence of classicality from quantum mechanics via decoherence \cite{Navascues:2010RSA, Kurzynski:2011XXX}, since we consider only a pair of (microscopic) objects, rather than an ensemble of objects. Our argument indicates that there is another option for a refinement of today's physics: we might be missing hidden post-quantum structures due to an ignorance of the correct measurement directions.

\begin{acknowledgments}
We thank O.\ G\"uhne, A.\ Ahlbrecht, and C.\ Budroni for helpful discussions 
and the Centro de Ciencias de Benasque, where part of this work has been done, 
for its hospitality during the workshop on quantum information $2012$.
This work has been supported by
 the Austrian Science Fund (FWF): Y376-N16 (START prize),
 the BMBF (CHIST-ERA network QUASAR),
 the EU (Marie-Curie CIG 293933/ENFOQI, Coquit, QFTCMPS),
V.B.S. is supported by an ETH postdoctoral fellowship and the SNF through the 
 National Centre of Competence in Research Quantum Science and Technology.
This work was
supported, in part, by the cluster of excellence EXC 201 ''Quantum Engineering and Space-Time Research'',
by the Deutsche Forschungsgemeinschaft (DFG).
\end{acknowledgments}

\bibliography{small}

\end{document}